\documentclass[11pt,letterpaper]{article}
\newcount\Colorflag
\usepackage[margin=1in]{geometry}

\usepackage{url}  %
\usepackage{graphicx}  %

\usepackage{amsmath,amsfonts,amssymb}
\usepackage{amsthm}
\newtheoremstyle{sltheorem}
{}                %
{}                %
{\slshape}        %
{}                %
{\bfseries}       %
{.}               %
{ }               %
{}                %
\theoremstyle{sltheorem}
\usepackage{bbm}
\usepackage{algorithm}
\usepackage{algpseudocode}
\usepackage[toc,page]{appendix}
\usepackage[hang,flushmargin]{footmisc}
\usepackage[inline]{enumitem}
\usepackage{multirow}
\usepackage{parskip}
\usepackage{breqn}
\usepackage{hyperref}
\usepackage[numbers]{natbib}

\DeclareMathOperator*{\E}{\mathbb{E}}
\let\Pr\relax
\DeclareMathOperator*{\Pr}{\mathbb{P}}
\def\reals{\mathbb{R}}

\newcommand{\POLY}{\text{POLY}}
\newcommand{\eps}{\varepsilon}
\newcommand{\SOLD}{\text{SOLD}}

\newcommand{\bS}{\mathbb{S}}

\newcommand{\bT}{\mathbb{T}}
\newcommand{\vB}{\mathbf{B}}

\newcommand{\vp}{\vec{p}}
\newcommand{\vv}{\vec{v}}

\newcommand{\calM}{\mathcal{M}}
\newcommand{\calC}{\mathcal{C}}
\newcommand{\calO}{\mathcal{O}}
\newcommand{\calR}{\mathcal{R}}

\newcommand{\bcalC}{\overline{\mathcal{C}}}

\newcommand{\calA}{\mathcal{A}}
\newcommand{\calQ}{\mathcal{Q}}
\newcommand{\calD}{\mathcal{D}}

\newcommand{\calB}{\mathcal{B}}
\newcommand{\calV}{\mathcal{V}}

\newcommand{\Rev}{\text{Rev}}

\newcommand{\LW}{\text{LW}}

\newcommand{\p}{\vec{p}}
\newcommand{\bbv}{\mathbf{\bar{v}}}
\newcommand{\bV}{\bar{V}}

\newcommand{\bv}{\bar{v}}

\newcommand{\halpha}{\hat{\alpha}}

\newcommand{\1}{\mathbbm{1}}
\newcommand\numberthis{\addtocounter{equation}{1}\tag{\theequation}}
\newcommand{\OPT}{\text{OPT}}
\newcommand{\bOPT}{\overline{\text{OPT}}}
\newcommand{\ALG}{\text{ALG}}

\newcommand{\calS}{\mathcal{S}}

\def\DQ{\mathrm{DQ}}
\def\BCDQ{\mathrm{BCDQ}}

\newcounter{dummy} \numberwithin{dummy}{section}

\newtheorem{theorem}[dummy]{Theorem}
\newtheorem{lemma}[dummy]{Lemma}
\newtheorem{definition}[dummy]{Definition}

\newtheorem{proposition}[dummy]{Proposition}

\begin{document}
\title{A Bridge between Liquid and Social Welfare in Combinatorial Auctions with Submodular Bidders}
\author{Dimitris Fotakis\thanks{Yahoo! Research, New York, NY USA and National Technical University of Athens, Athens, Greece, \tt{fotakis@\{oath.com;cs.ntua.gr\}}} 
\and Kyriakos Lotidis\thanks{National Technical University of Athens, Athens, Greece, \tt{klot@corelab.ntua.gr}}
\and Chara Podimata\thanks{Harvard University, Cambridge, MA, USA, \tt{podimata@g.harvard.edu}}
}

\maketitle
\begin{abstract}
We study incentive compatible mechanisms for Combinatorial Auctions where the bidders have submodular (or XOS) valuations and are budget-constrained. Our objective is to maximize the \emph{liquid welfare}, a notion of efficiency for budget-constrained bidders introduced by Dobzinski and Paes Leme (2014). We show that some of the known truthful mechanisms that best-approximate the social welfare for Combinatorial Auctions with submodular bidders through demand query oracles can be adapted, so that they retain truthfulness and achieve asymptotically the same approximation guarantees for the liquid welfare. More specifically, for the problem of optimizing the liquid welfare in Combinatorial Auctions with submodular bidders, we obtain a universally truthful randomized $O(\log m)$-approximate mechanism, where $m$ is the number of items, by adapting the mechanism of Krysta and V{\"o}cking (2012). 

Additionally, motivated by large market assumptions often used in mechanism design, we introduce a notion of competitive markets and show that in such markets, liquid welfare can be approximated within a constant factor by a randomized universally truthful mechanism. Finally, in the Bayesian setting, we obtain a truthful $O(1)$-approximate mechanism for the case where bidder valuations are generated as independent samples from a known distribution, by adapting the results of Feldman, Gravin and Lucier (2014).
\end{abstract}

\section{Introduction}\label{sec:intro}

Imagine that you are a social planner wanting to auction-off the seats of a local stadium at an extremely wealthy neighborhood (i.e., people have \emph{no} budget constraints for the seats) for a big concert. As a social planner, your goal is to allocate the seats in a way that maximizes (or, at least, approximates as closely as possible) the happiness of the people interested in these seats. However, different people have different seat preferences; some people are happy with two consecutive seats anywhere in the stadium, and some might want a whole row. Phrased in mechanism design language, this is a \emph{Combinatorial Auction}, where you seek to optimize the \emph{social welfare} by a truthful mechanism. Combinatorial Auctions, like the one above, appear in many AI-centric contexts (e.g., spectrum auctions, network routing auctions \citep{HS01}, airport time-slot auctions \citep{RSB82}, etc.) and have been a central topic in the study of Multi-Agent Systems. They have also experienced a recent interest in the AI community with works employing ML algorithms to overcome standard complexity problems (e.g., \citep{BL18,BLS18}). 

As if this problem was not hard enough to solve, imagine that you find out two unfortunate events; the stadium is in fact at a working-middle class neighborhood (i.e., people \emph{do} have budget constraints) and your boss is concerned about the effect of these budget constraints on the potential revenue. Now, the objective function should balance between the willingness and the ability of the people to pay for their seats. Motivated by usual discrepancies between the auction participants' ability and willingness to pay, \citet{DPL14} introduced the notion of \emph{liquid welfare}, which is the minimum of an agent's budget and valuation for a bundle of goods. As such, maximizing the liquid welfare achieves a reasonable compromise between social efficiency and potential for revenue extraction (which is constrained by the budgets). 

\smallskip\noindent{\bf Problem Definition.} 
More formally, a Combinatorial Auction (CA) consists of a set $U$ of $m$ items to be allocated to $n$ bidders. Each bidder $i$ has a valuation function $v_i : 2^U \to \reals_{\geq 0}$. Valuation functions, $v$, are assumed to be non-decreasing, i.e., $v(S) \leq v(T)$, for all $S \subseteq T \subseteq U$, and normalized $v(\emptyset) = 0$. For the objective of social welfare (SW), the goal is to compute a partitioning $\mathcal{S} = (S_1, \ldots, S_n)$ of the set of items, $U$, that maximizes $v(\mathcal{S}) = \sum_{i=1}^n v_i(S_i)$. For the objective of liquid welfare (LW), we assume that each bidder $i$ also has a budget $B_i \in \reals_{\geq 0}$ and the liquid welfare that can be extracted from agent $i$ for each set of items $S \subseteq U$ is $\bar{v}_i(S) = \min\{ v_i(S), B_i \}$\footnote{Slightly abusing the terminology, we refer to $\bv_i(S)$ as agent $i$'s \emph{liquid valuation}.}. Under this objective, the goal is to compute a partitioning $\mathcal{S} = (S_1, \ldots, S_n)$ of $U$ that maximizes $\bar{v}(\mathcal{S}) = \sum_{i=1}^n \bar{v}_i(S_i)$. 

We focus on CAs with submodular or XOS bidders. A set function $v: 2^U \to \reals_{\geq 0}$ is \emph{submodular} if for every $S, T \subseteq U$, $v(S) + v(T) \geq v(S \cap T) + v(S \cup T)$ and \emph{subadditive} if $v(S) + v(T) \geq v(S \cup T)$. A set function $v$ is \emph{XOS} (a.k.a. \emph{fractionally subadditive}, see \citep{F09}) if there exist additive functions $w_k  : 2^U \to \reals_{\geq 0}$ such that for every $S \subseteq U$, $v(S) = \max_{k} \{ w_k(S) \}$. The class of submodular functions is a proper subset of the class of XOS functions, which in turn is a proper subset of the class of subadditive functions.

Since bidder valuations have exponential size, a polynomial (in $m$ and $n$) algorithm must have oracle access to them. A \emph{value query} specifies a set $S \subseteq U$ and receives the value $v(S)$. A \emph{demand query}, denoted by $\DQ(v, U, \vec{p})$, specifies a valuation function $v$, a set $U$ of available items and a price $p_j$ for each available item $j \in U$, and receives the set (or bundle) $S \subseteq U$ maximizing $v(S) - \sum_{j \in S} p_j$, i.e., the set of available items that maximizes bidder's utility at these prices. For brevity, we often write $p(S) = \sum_{j \in S} p_j$ to denote the price of a bundle $S$. Demand queries are strictly more powerful than value queries. Value queries can be simulated by polynomially many demand queries, and in terms of communication cost, demand queries are exponentially stronger than value queries \citep{BN09}. Our mechanisms are polynomial-time, given access to demand oracles, which in general can be NP-hard to compute.

\subsection{Previous Work on Social Welfare}

Truthful maximization of SW in CAs with submodular or XOS bidders has been a central problem in Algorithmic Mechanism Design, with many powerful results. Due to space restrictions, we only discuss results most relevant to our work. While discussing previous work below, we assume XOS bidders and polynomial-time randomized truthful mechanisms that approximate the SW, by accessing valuations through demand queries, unless mentioned otherwise. 

In the \emph{worst-case} setting, where we do not make any further assumptions on bidder valuations, \citet{DNS06} presented the first truthful mechanism with a non-trivial approximation guarantee of $O(\log^2 m)$. \citet{D07} improved the approximation ratio to $O(\log m \log \log m)$ for the more general class of subadditive valuations. Subsequently, \citet{KV12} provided an elegant randomized online mechanism that achieves an approximation ratio of $O(\log m)$ for XOS valuations. \citet{D16} broke the logarithmic barrier for XOS valuations, by providing an approximation guarantee of $O(\sqrt{ \log m })$. We highlight that accessing valuations through demand queries is essential for these strong positive results. \citet{D11} proved that any truthful mechanism for submodular CAs with approximation ratio better than $m^{\frac{1}{2}-\eps}$ must use \emph{exponentially} many value queries. 

In the \emph{Bayesian} setting, bidder valuations are drawn as independent samples from a known distribution. \citet{FGL14} showed how to obtain item prices that provide a constant approximation ratio for XOS valuations. These results were significantly extended and strengthened in the recent work of \citet{DFKL17}, and a (truly) polynomial algorithm was provided as well.

\subsection{Intuition, Main Ideas, and Contribution}

Our aim is to extend these results to the objective of LW. To this end, we exploit the fact that most of the mechanisms above follow a simple pattern: first, by exploring either part of the instance in \cite{KV12} or the knowledge about the valuation distribution in \cite{FGL14}, the mechanism computes appropriate (a.k.a. \emph{supporting}) prices for all items. Then, these prices are ``posted'' to the bidders, who arrive one-by-one and select their utility-maximizing bundle, through a demand query, from the set of available items (see Algorithm~\ref{algo:pp}). 

The technical intuition behind the high level approach above is nicely explained in \cite[Section~1.2]{D16}. Let $\calO = (O_1, \ldots, O_n)$ be an optimal solution for the SW (in fact, any constant factor approximation suffices). The supporting price of item $i$ in $\mathcal{O}$ is $q_j = w_k(\{ j \})$, where $w_k$ is the additive valuation determining the value $v_i(O_i)$ (recall that valuation functions are XOS). Intuitively, $q_j$ is how much item $j$ contributes to the social welfare of $\mathcal{O}$. Then, a price of $p_j = q_j/2$ for each item $j$ is appropriate in the sense that a constant approximation to $v(\mathcal{O})$ can be obtained by letting the bidders arrive one-by-one, in an arbitrary order, and allocating to each bidder $i$ her utility maximizing bundle, chosen from the set of available items by a demand query (see \cite[Lemma~4.2]{D16}). 

Hence, approximating the SW by demand queries boils down to computing such prices $p_j$. In the Bayesian setting, prices $p_j$ can be obtained by drawing $n$ samples from the valuation distribution and computing the expected contribution of each item $j$ to a constant factor approximation of the optimal allocation (see Section~3 and Lemma~3.4 in \cite{FGL14}). Similarly, the idea of estimating the contribution of the items would work under some market uniformity assumption, as the one introduced in Definition~\ref{def:competitive}. In the worst-case setting, if we assume integral and polynomially-bounded valuations (i.e., that $\max_{i} \{ v_i(U) \} \leq m^d$, for some constant $d$), a uniform price for all items selected at random from $1, 2, 4, 8, \ldots, 2^{d \log m}$ results in an logarithmic approximation ratio. \citet{KV12} show how to estimate supporting prices online, by combining binary search and randomized rounding. Importantly, as long as each bidder does not affect the prices offered to her, this general approach results in (randomized, universally) truthful mechanisms.

Towards extending the above approach and results to the LW, our first observation (Lemma~\ref{lem:val-liqval}) is that if a valuation function $v$ is submodular (resp. XOS), then the corresponding liquid valuation function $\bar{v} = \min\{ v, B \}$ is also submodular (resp. XOS). Then, one can directly use the mechanisms of e.g., \citep{KV12,D16,FGL14} with valuation functions $\bar{v} = \min\{ v, B \}$ and demand queries of the form: $\DQ(\min\{v, B\}, U, \vec{p})$ (i.e., wrt. the liquid valuation of the bidders) and obtain the same approximation guarantees but now for the LW. However, the resulting mechanisms are no longer truthful; bidders still seek to maximize their \emph{utility} (i.e., value minus price) from the bundle that they get, subject to their budget constraint, rather than their \emph{liquid utility} (i.e., liquid value minus price). Specifically, given a set of items $U$ available at prices $p_j$, $j \in U$, a budget-constrained bidder $i$ wants to receive the bundle $S_i = \arg\max_{S \subseteq U} \{ v_i(S) - p(S)\,|\,p(S) \leq B_i \}$, and might not be happy with the bundle $S'_i = \arg\max_{S \subseteq U} \{ \bar{v}_i(S) - p(S) \}$ computed by the demand query for the liquid valuation%
\footnote{For a concrete example, consider a bidder with budget $B = 2$ and two items $a$ and $b$ available at prices $p_a = 2$ and $p_b = 1$. Assume that the bidder's valuation function is $v(\{ a \}) = v(\{ a, b \}) = 10$, $v(\{ b \}) = 2$ (and therefore, her liquid valuation is $\bar{v}(\{ a \}) =\bar{v}(\{ b \}) = \bar{v}(\{ a, b \}) = 2$). The bidder wants to get item $a$ at price $2$, which gives her utility $8$. However, the demand query for her liquid valuation function $\bar{v}$ allocates item $b$, which gives her a utility of $1$. Clearly, in this example, the bidder would have incentive to misreport her preferences to the demand query oracle.}. 

To restore truthfulness, we replace demand queries with \emph{budget-constrained demand queries}. A budget-constrained demand query, denoted by $\BCDQ(v, U, \vec{p}, B)$, specifies a valuation function $v$, a set of available items $U$, a price $p_j$ for each $j \in U$ and a budget $B$, and receives the set $S \subseteq U$ maximizing $v(S) - p(S)$, subject to $p(S) \leq B$, i.e., the set of available items that maximizes bidder's utility subject to her budget constraint. 

To establish the approximation ratio, we first observe that the fact that liquid valuations are XOS suffices for estimating supporting prices, as in previous work on the SW. Additionally, we show that the bundles allocated by $\BCDQ(v, U, \vec{p}, B)$ approximately satisfy the efficiency guarantees on the liquid welfare and the liquid utility of the allocated bundles (see Lemma~\ref{lem:bfdo}). Specifically, we observe that the approximation guarantees of mechanisms for the SW mostly follow from the fact that a demand query $\DQ(v, U, \vec{p})$ guarantees that for the allocated bundle $S$ and for any $T \subseteq U$, 
\begin{enumerate*}[label=(\roman*)]
\item $v(S) - p(S) \geq v(T) - p(T)$, and  
\item $v(S) \geq v(T) - p(T)$
\end{enumerate*}. In Lemma~\ref{lem:bfdo}, we show that a budget-constrained demand query, $\BCDQ(v, U, \vec{p}, B)$, guarantees that for the allocated bundle $S$ and any $T \subseteq U$,
\begin{enumerate*}[label=(\roman*)]
\item $2 \bar{v}(S) - p(S) \geq \bar{v}(T) - p(T)$, and
\item $\bar{v}(S) \geq \bar{v}(T) - p(T)$.
\end{enumerate*}
Using this property, we can prove the equivalent of \cite[Lemma~4.2]{D16}%
and also the approximation guarantees of the mechanisms in \citet{KV12,FGL14} but for the LW. 

\paragraph{Contribution.}
Formalizing the intuition above, we obtain 
\begin{enumerate*}[label=(\roman*)]
\item a randomized universally truthful mechanism that approximates the LW within a factor of $O(\log m)$ (Section~\ref{sec:worst-case}), and 
\item a posted-price mechanism that approximates the LW within a constant factor when bidder valuations are drawn as independent samples from a known distribution (Section~\ref{sec:stochastic}). 
\end{enumerate*} Both mechanisms assume XOS bidder valuations; the former is based on the mechanism of \citet{KV12} and the latter on the mechanism of \citet{FGL14}. 

Motivated by large market assumptions often used in Algorithmic Mechanism Design (see e.g., \cite{BCIMS05,EFV17,LX17} and the references therein), we introduce a competitive market assumption in Section~\ref{sec:lcma}. Competitive Markets are closer to practice, since they stand in between the stochastic and the worst-case settings, in terms of the assumptions made. The main idea is that when there is an abundance of bidders, even if we remove a random half of them, the optimal LW does not decrease by much. Then, computing supporting prices for all items based on a randomly chosen half of the bidders, and offering these prices through budget-constrained demand queries to the other half, yields a universally truthful mechanism that approximates LW within a constant factor (Theorem~\ref{thm:cm-apx}). 

Conceptually, in this work, we present a general approach through which known truthful approximations to the SW, that access valuations through demand queries, can be adapted so that they retain truthfulness and achieve similar approximation guarantees for the LW. The important properties required are that liquid valuation functions $\bar{v}$ belong to the same class as valuation functions $v$ (proven for submodular, XOS and subadditive valuations), and that the efficiency guarantees of budget-constrained demand queries on liquid welfare and liquid utility are similar to the corresponding efficiency guarantees of standard demand queries for liquid valuations (proven for all classes of valuations functions). Indeed, applying this approach to the mechanism of \citet{D16}, we obtain a universally truthful mechanism that approximates the LW for CAs with XOS bidders within a factor of $O(\sqrt{\log m})$ (the details are omitted due to space constraints). Similarly, we can take advantage of the improved results of \citet{DFKL17} in the Bayesian setting. All the missing proofs can be found in the full version of the paper on \citep{FLP18}.

\subsection{Previous Work on Liquid Welfare}

Liquid welfare was introduced as an efficiency measure for auctions when bidders are budget constrained in \citep{DPL14} (since it was known that getting any non-trivial approximation for the SW in these cases is impossible) and it corresponds to the optimal revenue an omniscient seller could extract from the set of the bidders, had he known their valuations and their budgets. Moreover, \citet{DPL14} proved a $O( \log n )$ (resp. $(\log^2 n)$)-approximation to the optimal LW for the case of a single divisible item and submodular (resp. subadditive) bidders. \citet{DPL14} and \citet{LX15} proved that the optimal LW can be approximated truthfully within constant factor for a single divisible good, additive bidder valuations and public budgets. Closer to our setting, \citet{LX17} provided a truthful mechanism that achieves a constant factor approximation to the LW for multi-item auctions with divisible items, under a large market assumption. Under similar large market assumptions, \citet{EFV17} obtained mechanisms that approximate the optimal revenue within a constant factor for multi-unit online auctions with divisible and indivisible items, and a mechanism that achieves a constant approximation to the optimal LW for general valuations over divisible items. However, prior to our work, there was no work on approximating the LW in CAs (in fact, that was one of the open problems in \cite{DPL14}). 

Our work is remotely related to the literature of \emph{budget feasible mechanism design}, a topic introduced by \citet{S10} and studied in e.g., \cite{DPS11,CGL11,BCGL12,BH16,WLL18}. Budget feasible mechanism design focuses on payment optimization in reverse auctions, a setting almost orthogonal to the setting we consider in this work.

\section{Notation and Preliminaries}\label{sec:model} 

The problem and most of the terminology and the notation are introduced in Section~\ref{sec:intro}. In this section, we introduce some additional notation required for the technical part. 

We use $\E[X]$ to denote the expectation of a random variable $X$ and $\Pr[E]$ to denote the probability of an event $E$. Let $\OPT$ (resp. $\bOPT$) denote the optimal SW (resp. LW)\footnote{The instance is clear from the context.}. For some $\rho > 1$, which may depend on $n$ and $m$, we say that a mechanism is $\rho$-approximate for the optimal SW (resp. LW) if it produces an allocation $\calS$ with $\rho \cdot v(\calS) \geq \OPT$ (resp. $\rho \cdot \bv(\calS) \geq \bOPT$).

Let a social choice function $f: \bV^n \to A$, which maps the set of liquid valuation functions of the bidders, $\bV: V \times B$, to an allocation, $A$, and a payment scheme $q = (q_1, \dots, q_n)$ for this allocation. A deterministic mechanism is defined by the pair $(f,q)$. Our mechanisms in this work are going to be \emph{randomized}, i.e., they are probability distributions over \emph{deterministic} mechanisms. The incentives desiderata for randomized mechanisms are usually either \emph{universal truthfulness} (when for all the deterministic mechanisms, the bidders' dominant strategy is truthfulness) or \emph{truthfulness in expectation} \citep{DFK10,DRY11} (when bidders' \emph{expected} utilities are maximized under truthful reporting of their private information). In this work, we are focusing on the former, stronger notion; the one of \emph{universal truthfulness}, under the bidders' budget constraints.

\begin{definition}[Universal Truthfulness under Budget Constraints] \label{def:truth-budgets}
Let $(\tilde{f}, \tilde{q})$ be a randomized mechanism over a set of deterministic mechanisms $\{(f^1, q^1 ), \dots, (f^k, q^k ) \}$. Mechanism $(\tilde{f}, \tilde{q})$ is \emph{universally truthful} if for all $i \in [n], \kappa \in [k]$ and for any $v_i'$ and any $B_i'$, such that $q^\kappa(v_i', v_{-i}) \leq B_i'$ and $q^\kappa(v_i, v_{-i}) \leq B_i$, it holds that:
\begin{dmath*}
v_i(f^\kappa(v_i,B_i, v_{-i},B_{-i})) - q^\kappa(v_i,B_i, v_{-i}, B_{-i}) \geq v_i(f^\kappa(v_i', B_i', v_{-i},B_{-i})) - q^\kappa(v_i', B_i', v_{-i}, B_{-i})
\end{dmath*}
\end{definition}

\section{Approach}\label{sec:approach}

First, we show that if the bidder valuations are submodular (resp. XOS, subadditive), then their liquid valuations are submodular (resp. XOS, subadditive) as well.

\begin{lemma}\label{lem:val-liqval}
Let $v$ be a non-negative monotone submodular (resp. XOS, subadditive) function. Then, for any $B \in \reals_{\geq 0}$, $\bv = \min\{ v, B \}$ is also monotone submodular (resp. XOS, subadditive).
\end{lemma}         

In Algorithm~\ref{algo:pp}, we present a universally truthful (since the prices offered to each bidder do not depend on her declaration and demand queries maximize bidders' utility) mechanism, which is a simplified version of the mechanism in \citep{KV12} for approximating SW in CAs. Since for the LW, bidders have budgets, we replace the demand queries $\DQ(v_i, U_i, \vec{p}^{(i)})$ in line~4 with budget constrained demand queries $\BCDQ(v_i, U_i, \vec{p}^{(i)}, B_i)$\footnote{In all our mechanisms, if budgets are larger than the valuations of the allocated bundles, the mechanism with $\BCDQ$ behaves identically to the mechanism with $\DQ$ (i.e., revenue and SW are not affected by the change in the objective.)}. As explained in Section~1.2, Algorithm~\ref{algo:pp} with $\BCDQ$s remains universally truthful for budget-constrained bidders.  

\begin{algorithm}[t]
\caption{Core Mechanism} \label{algo:pp}
\begin{algorithmic}[1]
\State Fix an ordering $\pi$ of bidders and set $U_1 = U$. 
\State Set initial prices for the items: $\vp^{(1)} = (p_1^{(1)}, \dots, p_m^{(1)})$. 
\For{each bidder $i = 1, \dots, n$ according to $\pi$}
\State Let $S_i = \DQ(v_i, U_i, \vp^{(i)},)$ 
\State With probability $q$, allocate $S_i$ to $i$ and set $U_{i+1} = U_i \setminus S_i$\,. Otherwise, set $U_{i+1} = U_i$\,.
\State Update item prices to $\vp^{(i+1)} = ( p_1^{(i+1)}, \dots, p_m^{(i+1)})$. 
\EndFor
\end{algorithmic}
\end{algorithm}

\begin{lemma}[Truthfulness of BCDQs]
For budget-constrained bidders, Algorithm~\ref{algo:pp} with $\BCDQ$s in line~4, is universally truthful. 
\end{lemma}

The lemma follows directly from Definition~\ref{def:truth-budgets}. Nevertheless, universal truthfulness is not our sole desideratum; in each of the three settings analyzed in the following sections, we show why mechanisms similar in spirit to Algorithm~\ref{algo:pp} with $\BCDQ$s, yield good approximation guarantees for the LW. Before the setting-specific analyses, we relate the efficiency of $\BCDQ$ to the efficiency of standard $\DQ$s for liquid valutions. 

\begin{lemma}\label{lem:bfdo}
Let $S \subseteq U$ be the set allocated by the BCDQ for a bidder with valuation $v$ and budget $B$. Then, for every other $T \subseteq U$, the following hold:
\begin{enumerate}[label=\roman*)]
\item \label{lem:prop1} $\bv(S) \geq \bv(T) - p(T)$
\item \label{lem:prop2} $2\bv(S) - p(S) \geq \bv(T) - p(T)$.
\end{enumerate}
\end{lemma}

\begin{proof} We will prove each claim of the lemma separately. For claim~\ref{lem:prop1}, notice that if $p(T) > B$, then the right hand side of the inequality will be negative and thus, the inequality trivially holds. So, we will focus on the case where $p(T) \leq B$. We distinguish the following cases: 
\begin{enumerate}
\item ($\bv(S) = v(S)$ and $\bv(T) = v(T)$) Hence, $B \geq v(T)$. Bundle $T$ was \emph{considered} at the time of the query and yet, the query returned set $S$. Thus: $\bv(S) \geq \bv(S) - p(S) = v(S) - p(S) \geq v(T) - p(T) = \bv(T) - p(T)$.  
\item ($\bv(S) = B$ and $\bv(T) = B$) Then, the inequality trivially holds since: $B \geq B - p(T)$ and prices are non-negative.
\item ($\bv(S) = B$ and $\bv(T) = v(T)$) The inequality holds since: $B \geq B - p(T) \geq v(T) - p(T) = \bv(T) - p(T)$.
\item ($\bv(S) = v(S)$ and $\bv(T) = B$) Hence, $B \leq v(T)$. Bundle $T$ was \emph{considered} at the time of the query and yet, the query returned set $S$. Thus, $\bv(S) \geq \bv(S) - p(S)   = v(S) - p(S) \geq v(T) - p(T) \geq B - p(T) = \bv(T) - p(T)$.
\end{enumerate}
This concludes our proof for claim~\ref{lem:prop1}.

For claim~\ref{lem:prop2}, notice that since $S$ is the set received from the BCDQ, then it is \emph{affordable}: $\bv(S) \geq p(S)$. Adding this inequality to the inequality of claim \ref{lem:prop1}, we have that: $2\bv(S) - p(S) \geq \bv(T) - p(T)$. \end{proof}

\section{Worst-Case Setting}\label{sec:worst-case}

For the worst-case instances, adapting appropriately our Core Mechanism, we present Algorithm~\ref{algo:kv12} (based again, on the mechanism of \citep{KV12}). The only difference is that budget-constrained bidders in Algorithm~\ref{algo:kv12} are restricted to using $\BCDQ$s, instead of $\DQ$s, thus making the mechanism universally truthful (see Section~\ref{sec:approach}). Resembling the analysis of \cite{KV12}, we show that for $1/q = \Theta(\log m)$, Algorithm~\ref{algo:kv12} achieves an approximation ratio of $O(\log m)$ for the LW. First, we note that parameter\footnote{$L$ can be computed with standard techniques, as explained in \citep{KV12}.} $L$ is selected so that there exists only one bidder whose liquid valuation for $U$ (weakly) exceeds it.  
\begin{algorithm}[t]
\caption{KV-Mechanism for Liquid Welfare} \label{algo:kv12}
\begin{algorithmic}[1]
\State Fix an ordering $\pi$ of bidders and set $U_1 = U$. 
\State Set initial prices $p_1^{(1)} = \cdots = p_m^{(1)} = \frac{L}{4m}$. 
\For{each bidder $i = 1, \dots, n$ according to $\pi$}
\State Let $S_i = \BCDQ(v_i, U_i, \vp^{(i)}, B_i)$ 
\State With probability $q$, allocate $R_i = S_i$ to $i$ and set $U_{i+1} = U_i \setminus S_i$\,. Otherwise, set $U_{i+1} = U_i, R_i = \emptyset$\,.
\State Update prices $\forall j \in S_i$: $p_j^{(i+1)} =2 p_j^{(i)}$. 
\EndFor
\end{algorithmic}
\end{algorithm}

\begin{theorem}\label{thm:lw-wc}
Algorithm~\ref{algo:kv12} is universally truthful and for $q = 1/\Theta(\log m)$, achieves an approximation ratio of $O(\log m)$ for the LW. 
\end{theorem}

We present a series of Lemmas that will lead us naturally to the proof of the Theorem. Let $\calS = (S_1, \dots, S_n)$ and $\calR = (R_1, \dots, R_n)$ the provisional and the final allocation of Algorithm~\ref{algo:kv12} respectively. First, we provide two useful bounds on $\bv(\calS)$. We find it important to also discuss an overselling variant of Algorithm~\ref{algo:kv12}. In the \emph{Overselling} variant, allow us to assume that for Step~5 of Algorithm~\ref{algo:kv12}, $q=1$ (i.e., $S_i$ is allocated to bidder $i$ with certainty) and $U_{i+1} = U_i = U$ (thus the name of the variant). The \emph{Overselling} variant allocates at most $k=\log(4m)+2$ copies of each item and collects a liquid welfare within a constant factor of the optimal LW. To see that, observe that for $q=1$, after allocating $k-1$ copies of some item $j$, $j$'s price becomes $\frac{L}{4m} 2^{\log(4m)+1} = 2L$. Then, there is only one agent with liquid valuation larger than $L$ who can get a copy of $j$. 

\begin{lemma}\label{lem:lb-1}
Let $p_j$ denote the final price of each item $j$. Then, for \emph{any} sets $U_1, \dots, U_n \subseteq U$ of items available when the bidders arrive, Algorithm~\ref{algo:kv12} with $q = 1$ satisfies $\bv(\calS) \geq \sum_{j \in U} p_j - \frac{L}{4}$.
\end{lemma}

\begin{proof}
Since bidders are individually rational and do not violate their budget constraints, for every bidder $i$ it holds that $B_i \geq \sum_{j \in S_i} p_j^{(i)}$ and $v_i(S_i) \geq \sum_{j \in S_i} p_j^{(i)}$. The rest of the proof is identical to the proof of \cite[Lemma~2]{KV12} for $b=1$. Specifically, let $\ell_j^{(i)}$ be the number of copies of item $j$ allocated just before bidder $i$ arrives, and let $\ell_j$ be the total number of copies of item $j$ allocated by Algorithm~\ref{algo:kv12} with $q = 1$. Then, using the fact that $p_j =L \cdot \frac{2^{\ell_j}}{4m}$: 
\begin{align*}
\bv(\calS) & \geq \sum_{i=1}^n \sum_{j \in S_i} p_j^{(i)} = \frac{L}{4m} \sum_{i=1}^n \sum_{j \in S_i} 2^{\ell_j^{(i)}} = \frac{L}{4m} \sum_{j \in U} (2^{\ell_j} - 1) = \sum_{j \in U} p_j - \frac{L}{4}
\end{align*}
\end{proof}

\begin{lemma}\label{lem:lb-2}
For sets $U_1 = \dots = U_n \subseteq U$, the \emph{Overselling} variant of Algorithm~\ref{algo:kv12} with $q = 1$ satisfies $\bv(\calS) \geq \bOPT - \sum_{j \in U} p_j$.
\end{lemma}

\begin{proof}
Let $\calO = (O_1, \dots, O_n)$ be the optimal allocation. From Lemma~\ref{lem:bfdo}, we get that for each bidder $i$, $\bv(S_i) \geq \bv(O_i) - \sum_{j \in O_i} p_j^{(i)} \geq \bv(O_i) - \sum_{j \in O_i} p_j$, where we use that the final price of each item is the largest one. Summing over all bidders, we have that $\bv(\calS) \geq \bv(\calO) - \sum_{i =1}^n \sum_{j \in O_i} p_j \geq \bOPT - \sum_{j \in U}p_j$, where the  last inequality uses the fact that the optimal solution is feasible and thus, each item is allocated at most once in $\calO$. 
\end{proof}

\begin{lemma}\label{lem:oversell-apx}
The \emph{Overselling} variant of Algorithm~\ref{algo:kv12} with $q = 1$ allocates at most $\log(4m)+2$ copies of each item and computes an allocation $\calS$ with liquid welfare $\bv(\calS) \geq \frac{3}{8} \bOPT$.
\end{lemma}

\begin{proof}
Follows directly from Lemma~\ref{lem:lb-1}, Lemma~\ref{lem:lb-2} and the fact that $\bOPT \geq L$. 
\end{proof}

Of course, the allocation  $\calS$ in Lemma~\ref{lem:oversell-apx} is infeasible, since it allocates a logarithmic number of copies of each item. The remedy is to use an allocation probability $q = 1/\Theta(\log m)$. For such values of $q$, we can plug in the proof of \cite[Lemma~6]{KV12} (which just uses that the valuation functions are fractionally subadditive) and show that for each agent $i$ and for all $A \subseteq U$, $\E[\bv_i(A \cap U_i)] \geq \bv_i(A)/2$. We are now ready to conclude the proof of Theorem~\ref{thm:lw-wc}. 
\begin{lemma}
For Algorithm~\ref{algo:kv12} with $q^{-1} = 4(\log(4m) + 1)$, it holds that $\E[\bv (\calS)] \geq \bOPT/8$ and $\E[\bv(\calR)] \geq q\,\bOPT/8$. 
\end{lemma}

\begin{proof}
Let $\calO = ( O_1, \ldots, O_n )$ be the optimal allocation. For each bidder $i$, Lemma~\ref{lem:bfdo} implies that the response $S_i$ of $\BCDQ$ satisfies $\bv_i(S_i) \geq \bv_i(O_i \cap U_i) - \sum_{j \in O_i \cap U_i}p_j^{(i)}$, for any $U_i$ resulted from the outcome of the random coin flips. Therefore, $\E [\bv_i(S_i) ] \geq \E [ \bv_i(O_i \cap U_i) ] - \E [\sum_{j \in O_i \cap U_i} p_j^{(i)}]$. By the choice of $q$, for any bidder $i$, $\E[ \bv_i ( O_i \cap U_i )] \geq  \bv_i(O_i)/2$. Then, working with the expectations as in the proofs of Lemma~\ref{lem:lb-1}, Lemma~\ref{lem:lb-2} and Lemma~\ref{lem:oversell-apx}, we can show that $\E[\bv (\calS)] \geq \bOPT/8$. Finally, one can use linearity of expectation and show that $\E[\bv(\calR)] = q \E[\bv(\calS)]$. The details are omitted, due to lack of space, and can be found in \cite[Lemma~4]{KV12}.
\end{proof}

\section{Competitive Market}\label{sec:lcma}

\citet{BCIMS05} were the first ones to define a \emph{budget dominance parameter} that corresponded to the ratio of the maximum budget of all the bidders to the value of the optimum SW in the context of multi-unit auctions with budget-constrained bidders. More recently, \citet{EFV17} and \citet{LX17} used similar notions of budget dominance\footnote{Namely, $\forall i \in [n]: B_i \leq \bOPT/(mc)$, with $c$ a large constant.} (termed \emph{large market assumptions}) as a means to achieve constant factor approximation to the LW in multi-unit auctions and auctions with divisible items respectively. However, for the case of \emph{indivisible} items, it is clear that the definition of a large market used in the previous cases, becomes almost void (see Appendix~\ref{appendix:lcma} for a discussion). Below, we first introduce our definition of Competitive Markets for indivisible goods and then, show how one can obtain a constant factor approximation of the optimal LW, when bidders have XOS liquid valuations.

\begin{definition}[$(\eps, \delta)$ - Competitive Market]\label{def:competitive}
Let $0 \leq \eps < 2$ and a constant $0 \leq \delta \leq 1/2$. A market is called $(\eps, \delta)$ - Competitive Market, if for any randomly removed set of bidders, $\bS$, with cardinality $n/2$, then for the remaining set of bidders, $\bT$, it holds that: 
\begin{equation}
\Pr\!\left[\bOPT_{\bT} \geq \left(1 - \frac{\eps}{2}\right) \cdot \bOPT \right] \geq 1 - \delta 
\end{equation}
where by $\bOPT_{\bT}$ we denote the optimal LW achieved by bidders in set $\bT$. 
\end{definition}

\begin{proposition}
In an $(\eps,\delta)$ - Competitive Market, let $\bS \subseteq [n]$ be randomly chosen s.t. $|\bS| = \frac{n}{2}$ and let $\bT = [n] \setminus \bS$. Then:
\begin{align*}
\Pr\!&\left[\left\{\bOPT_{\bT} \geq \left(1-\tfrac{\eps}{2}\right)\bOPT\right\}\cap \left\{\bOPT_{\bS} \geq \left(1- \tfrac{\eps}{2}\right)\bOPT\right\}\right] \\ 
&\geq 1 - 2\delta
\end{align*}
\end{proposition}

\begin{proof}
Let $X_{\bS}$ the event that $\bOPT_{\bS} \geq \left(1-\frac{\eps}{2}\right)\bOPT$ and $X_{\bT}$ the event that $\bOPT_{\bT} \geq \left(1-\frac{\eps}{2}\right)\bOPT$. Then, we have: 
\begin{equation*}
\Pr \left[ X_{\bS} \cap X_{\bT} \right] = 1 - \Pr \left[ \overline{X_{\bS}} \cup \overline{X_{\bT}} \right] \geq 1 - 2\delta  
\end{equation*} 
where the inequality follows from the Union Bound.
\end{proof}

We are now ready to state our Competitive Market mechanism that will be used for approximating the optimal LW. We note here that the greedy algorithm $\calA$ is due to \citet{LLN06}. 

\begin{algorithm}[htbp]
\caption{Competitive Market (CM) Algorithm}\label{algo:lcma}
\begin{algorithmic}[1]
\State Divide the bidders into sets $\bS, \bT$ uniformly at random, s.t., $|\bS| = \frac{n}{2} = |\bT|$.
\State Run the greedy algorithm $\calA$ for bidders in $\bS$ and denote the solution obtained by $\calA^{\bS}$.
\For{$j \in U$}
\State Set $p_j = \frac{1}{2\beta}\bv \left( \calA_j^{\bS} \right)$, where $\beta > 1$ is a constant
\EndFor
\State Fix an internal ordering of bidders in $\bT$, $\pi$, and set $U_1 = U$.
\For{each bidder $i \in \bT$ arriving according to $\pi$}
\State Let $S_i = \BCDQ(v_i, U_i, \vp)$.
\State Set $U_{i+1} = U_i \setminus S_i$.
\EndFor
\end{algorithmic}
\end{algorithm}

As usual, we denote $\calS = (S_1, \dots, S_n)$ the final allocation from mechanism presented in Algorithm~\ref{algo:lcma}. Valuations of bidders are XOS (and so are the liquid valuations (Lemma \ref{lem:val-liqval})); let $a_i$ be the maximizing clause of $S_i$ in the liquid valuation $\bv_i$ of bidder $i$. Since $a_i$'s are additive, for each bidder $i$ and $j \in S_i$ let $q_j = a_i (\{j\})$. Notice that $\sum_{i \in [n]} \bv (S_i) = \sum_{j \in \cup_{i \in [n]} S_i}q_j$. We denote by $\bOPT_{\bT} = \sum_{j \in U} q_j^{\bT}$, where $q_j^{\bT}$ is the contribution of item $j$ in $\bOPT_{\bT}$. We divide the set of all items $U$ into two sets; the set of \emph{competitive} items, denoted by $\calC$ and the set of \emph{non-competitive} items, denoted by $\bcalC = \calM \setminus \calC$. The following lemma upper bounds the contribution of non-competitive items in the optimal solution.

\begin{lemma}\label{lem:contr-non-compet}
Let $\calC = \{j \big| q_j^{\bT} > \frac{\bv(\calA_j^{\bS})}{\beta}\}$ for constant $\beta >1$. Then, $\sum_{j \in \bcalC} q_j^{\bT} \leq \frac{\eps}{2(\beta - 1)}\bOPT$ and $\sum_{j \in \calC}q_j^{\bT} \geq \frac{\beta(2-\eps)-2}{2(\beta - 1)}\bOPT$.
\end{lemma}

\begin{proof}
From Definition \ref{def:competitive}, it holds with constant probability (w.c.p) that: 
\begin{equation*}
\bOPT \geq \sum_{j \in \calC}q_j^{\bT} + \sum _{j \in \bcalC}q_j^{\bT} =  \sum _{j \in U}q_j^{\bT} \geq \left(1-\frac{\eps}{2}\right) \cdot \bOPT 
\end{equation*}
Let $\bS_{\bcalC} \subseteq \bS$ be the set of the bidders that are allocated the non-competitive items from the greedy algorithm $\calA$ when running on set $\bS$. Then, in the augmented set $\bT \cup \bS_{\bcalC}$, there exists an allocation $\calQ$%
\footnote{Allocation $\calQ$ is realized by allocating all items in $\calC$ to bidders in $\bT$ that also had them in the $\bOPT_{\bT}$ allocation and all items in $\bcalC$ to the bidders in $\bS_{\bcalC}$ that had them in the allocation of the greedy $\calA$. The claim is completed by submodularity.} 
with liquid valuation, 
\begin{equation}\label{eq:lv1}
\bv(\calQ) \geq \sum_{j \in \calC} q_j^{\bT} + \sum_{j \in \bcalC} \bv\left(\calA_j^{\bS}\right)
\end{equation}
and therefore we have w.c.p:
\begin{align*}
\bOPT &\geq \bv(\calQ) \geq \sum_{j \in \calC} q_j^{\bT} + \sum_{j \in \bcalC} \bv \left(\calA_j^{\bS}\right) \geq  \sum_{j \in \calC}q_j^{\bT} + \beta \sum _{j \in \bcalC}q_j^{\bT} \\
&\geq \left(1 -\frac{\eps}{2}\right)\bOPT + (\beta -1)\sum _{j \in \bcalC}q_j^{\bT} 
\end{align*}
Re-arranging the latter and using the fact that 
\begin{equation*}
\sum_{j \in \calC} q_j + \frac{\eps}{2(\beta - 1)} \bOPT \geq \sum_{j \in U}q_j^{\bT} \geq \left(1-\frac{\eps}{2}\right)\bOPT
\end{equation*} 
then, for the items in $\calC$ it holds w.c.p that: $\sum_{j \in \calC}q_j^{\bT} \geq \frac{\beta(2-\eps)-2}{2(\beta - 1)}\bOPT$.
\end{proof}

In the next Lemma, we prove a lower bound on the contribution of competitive items to the solution obtained by the greedy algorithm, with respect to $\bOPT$.

\begin{lemma}\label{lem:bound-alg}
$\sum_{j \in \calC} \bv \left( \calA_j^{\bS} \right) \geq \frac{2 (\beta - 1) - \eps \cdot (3\beta -1)}{4(\beta-1)}\bOPT$.
\end{lemma}

\begin{proof}
Combining Inequality~\eqref{eq:lv1} and Lemma~\ref{lem:contr-non-compet} we get that $\sum_{j \in \bcalC} \bv \left( \calA_j^{\bS} \right) \leq \frac{\beta \eps}{2 (\beta-1)} \bOPT$. Algorithm $\calA$ provides a $2$-approximation to the optimal LW of set $\bS$ \citep{LLN06}, so w.c.p we have: 
\begin{equation*}
\sum_{j \in \calC} \bv \left( \calA_j^{\bS} \right) + \sum_{j \in \bcalC} \bv \left( \calA_j^{\bS} \right) \geq \frac{1}{2} \bOPT_{\bS} \geq \frac{1 - \frac{\eps}{2}}{2} \bOPT
\end{equation*}
Combining the last two equations, we get the result.
\end{proof}

\begin{theorem}\label{thm:cm-apx}
The CM Algorithm is \emph{universally truthful} and achieves, in expectation, a constant approximation to the optimal LW, i.e.,
\begin{equation*}
\E \left[ \bv \left( \calS \right) \right] \geq (1 - 2\delta) \cdot \frac{2(\beta -1)-\eps\cdot(3\beta-1)}{16\beta(\beta-1)}\bOPT
\end{equation*}
\end{theorem}
\begin{proof}
Since the bidders that control the prices being posted belong to set $\bS$ and they never get any item, it is their (weakly) dominant strategy to report their valuations and their budgets truthfully. Furthermore, the bidders that are buying under the said posted prices belong to set $\bT$ and they make BCDQs, which we showed to be truthful. Finally, the bidders are \emph{uniformly at random} split to sets $\bS$ and $\bT$.

For each item $j \in \calC$ we have $q_j^{\bT} > \bv (\calA_j^{\bS})/\beta$. Therefore, there exists an allocation for bidders in $\bT$ and items in $\calC$ that is supported by prices $p_1, \dots, p_m$, where $p_j =\bv(\calA_j^{\bS})/\beta$. Thus, a modification of \cite[Lemma~4.2]{D16} implies that if we we set $p_j' = p_j/2$, for each $j \in \calC$, and run a fixed price auction in $\bT$ with prices $p_1',\dots,p_m'$, we get that $\bv(\calS) \geq \sum_{j \in \calC} p_j/4$. Using the latter, along with the prices of the items, we have that 
\begin{equation*}
\bv(\calS) = \frac{1}{4\beta}\sum_{j \in \calC} \bv(\calA_j^{\bS})\geq \frac{2(\beta -1)-\eps(3\beta-1)}{16\beta(\beta-1)} \bOPT
\end{equation*}
where the last inequality is due to Lemma \ref{lem:bound-alg}. Thus, we conclude that $\E\left[\bv\left( \calS \right)\right] \geq \left(1-2\delta\right)\frac{2(\beta -1)-\eps \cdot (3\beta-1)}{16\beta(\beta-1)}\bOPT$.
\end{proof}

\section{Bayesian Setting}\label{sec:stochastic}

The Bayesian Setting offers a great middle ground between the unstructured worst-case instances and the very structured Competitive Markets. In this setting, let $\vv = (v_1, \dots, v_n)$ be a profile of bidder valuations and $\vB = (B_1, \dots, B_n)$ a profile of bidder budgets. Assume that the bidders' valuations are drawn independently from distributions $\calV_1, \dots, \calV_n$ and the budgets from distributions $\calB_1, \dots, \calB_n$. For simplicity, let us assume that their liquid valuations are drawn independently from distributions $\calD_1, \dots, \calD_n$. We will denote by $\calD = \calD_1 \times \dots \times \calD_n$ the product distribution where liquid valuations profiles, $\bbv = (\bv_1, \dots, \bv_n)$, are independently drawn from.  

We are going to show that the results presented in \citet{FGL14} can be extended for budget-constrained bidders. Specifically, we are going to show that, if liquid valuations are fractionally subadditive, then we can create appropriate prices such that, when presented to the bidders in a posted-price mechanism and bidders are making BCDQs, then we can obtain universally truthful constant-factor approximation mechanisms for the LW in Bayesian CAs. Our Lemma \ref{lem:like3.4} establishes the existence of such appropriately scaled prices. The key component activating our results is that instead of reasoning about the \emph{utility} achieved from the bundle purchased by bidder $i$ (as received by the BCDQ), we instead have to use Lemma \ref{lem:bfdo}. %
We also note that using our techniques one could even achieve the better approximation guarantees presented by \citet{DFKL17}. Their analysis is significantly more complex, however, and we omit it in the interest of space. %

\begin{theorem}\label{thm:anon-price}
Let distribution $\calD$ over XOS liquid valuation profiles be given via a sample access to $\calD$. Suppose that for every $\bbv \sim \calD$, we have: 
\begin{enumerate*}[label=\roman*)]
\item black-box access to a LW maximization algorithm, $\ALG$\footnote{$\ALG$ can be any algorithm that provides a $O(1)$-approximation to the optimal LW, since we do not care about incentives (access to $\ALG$ will only happen for ghost samples).}  
\item an XOS value query oracle (for liquid valuations sampled from $\calD$)\footnote{An XOS value oracle takes as input a set $T$ and returns the corresponding additive representative function for the set $T$, i.e., an additive function $A_i(\cdot)$, such that (i) $\bv_i(S) \geq A_i(\hat{S})$ for any $\hat{S} \subset [m]$ and (ii) $\bv_i(T) = A_i(T)$.}.
\end{enumerate*}
Then, for any $\epsilon > 0$, we can compute item prices in $\POLY(m,n,1/\epsilon)$ time such that, for any bidder arrival order, the expected liquid welfare of the posted price mechanism is at least $\frac{1}{4}\E_{\bbv \sim \calD} [\bbv (\ALG (\bbv))] - \epsilon$, where by $\ALG(\bbv)$ we denote the solution produced by algorithm $\ALG$.
\end{theorem}

\begin{lemma}\label{lem:like3.4}
Given a distribution $\calD$ over XOS liquid valuations, let $\p$ be the price vector s.t. $p_j = \frac{1}{2} \E_{\bbv \sim \calD} [\LW_j(\bbv)]$. Let $\p'$ be any price vector such that $|p_j' - p_j| < \delta$ for all $j \in [m]$. Then, for any arrival order, $\pi$, bidders buying bundles by making BCDQs under prices $\p'$ results in expected liquid welfare at least $\frac{1}{4} \E_{\bbv \sim \calD} [\bbv (\ALG(\bbv))] - \frac{m\delta}{2}$.
\end{lemma}

\section{Conclusion}\label{sec:concl}

In real-life auctions, bidders are \emph{always} constrained by budgets, which we tend to overlook due to the technical difficulties that they add. The role of budgets in welfare/revenue optimization is amplified in CAs, where bidders have richer valuations and hence, studying budgeted CAs is a step towards bridging the gap between the theory on truthful mechanism design for CAs and constraints faced in practice. In this work, we showed how the liquid welfare can be approximated in CAs where bidders are budget-constrained in three settings: worst-case, Competitive Markets and stochastic. The most meaningful question that arises from our work (apart, of course, from the ever existent one of lowering the approximation guarantee in worst-case instances) is related to the competitive markets. We conjecture that the condition that we provide can be made even weaker, and leave it to future research. 

Finally, our results can also be used to extend a variety of already known results in CAs without budgets, to CAs \emph{with} budget-constrained bidders. For example, Lemma~\ref{lem:bfdo} (with some changes in the constants of \citep{DK17}) implies a constant factor approximation for best response dynamics in XOS CAs with budgeted bidders, that apply after a single round of bid updates. 

\section{Acknowledgments}

We are thankful to Yiling Chen for her most valuable feedback in an initial draft of this work, as well as the anonymous reviewers for their helpful comments and suggestions. Part of this work was done while Chara Podimata was visiting Yahoo Research, NY. This work is partially supported by the National Science Foundation under grant CCF-1718549 and the Harvard Data Science Initiative. Any opinions, findings, conclusions, or recommendations expressed here are those of the authors alone.

\bibliographystyle{aaai} 
\bibliography{refs}

\newpage
{\Large {\bf Appendix}}

\appendix

\section{Supplementary Material for Section \ref{sec:approach}.}\label{appendix:approach}

\begin{proof}[Proof of Lemma \ref{lem:val-liqval}]
Clearly, capping valuation with budget does not affect monotonicity. We provide the proof for each case (i.e., submodular, XOS, subbaditive) separately.
\begin{itemize}
\item (submodular) Let $v$ be a monotone submodular set function. Then, by the definition of submodularity, for sets $T \subseteq S$ and $j \notin S$ we have: 
\begin{equation}\label{eq:def-submod}
v \left( S \cup \{j\} \right) - v \left( S \right) \leq v \left( T \cup \{j\} \right) - v \left( T \right)
\end{equation}
Further, since $v$ is monotone: $v(T) \leq v(S)$, which implies that $\bv(T) \leq \bv(S)$. We distinguish the following cases: 
\begin{enumerate}
\item If $B \leq v \left( T \cup \{j\} \right) \leq v \left(S \cup \{j\} \right)$. Then, for the liquid valuations we have: $\bv\left( S \cup \{j\} \right) - \bv\left(S \right) = B - \bv\left(S \right) \leq B - \bv(T) \leq \bv \left( T \cup \{j\} \right) - \bv(T)$, where the first inequality is due to monotonicity.
\item If $\bv( T \cup \{j\}) \leq \bv (S \cup \{j\}) \leq B$. Then, $\bv( S \cup \{j\} ) - \bv(S ) = v ( S \cup \{j\} ) - v(S) \leq v ( T \cup \{j\} ) - v(T) = \bv(T \cup \{j\}) - \bv(T)$.
\item If $v ( T \cup \{j\} ) \leq B \leq v ( S \cup \{j\} )$. This breaks down to the following two cases; on the one hand, if $v(S) \geq B$ then, $\bv( S \cup \{j\} ) - \bv(S ) = 0 \leq v(T \cup \{j\}) - v(T) = \bv(T \cup \{j\}) - \bv(T)$. On the other hand, if $v(S) < B$, then $\bv(S \cup \{j\}) - \bv(S) = B - v(S) \leq v(S \cup \{j\}) - v(S) \leq v(T \cup \{j\}) - v(T) = \bv(T \cup \{j\}) - \bv(T)$. Finally, we remark that due to monotonicity, these cases are the only possible ones. 
\end{enumerate}

\item (XOS) Let $v$ be an XOS set function; there exist additive functions $\alpha_1, \dots, \alpha_l$ s.t. $v(S) = \max_{i \in [l]} \alpha_i (S)$. In order for $\bv$ to XOS, we need to prove that there exist additive functions $\alpha_1', \dots, \alpha_l'$ s.t. $\bv(S) = \max_{i \in [l]} \alpha_i'(S)$. For each function $\alpha_i$ we are going to define $m!$ functions, one for each permutation $\pi$ of the items. Suppose a specific ordering $\pi_t$ of the items $\{1,2,\dots,m\}$ and let $\pi_t(j)$ be the position of item $j$ in ordering $\pi_t$.  We define $ \beta_i^{\pi_t}$ as: $\beta_i^{\pi_t}(\{j\}) = \alpha_i(\{j\})$, if $\sum_{k: \pi_t(k) \leq \pi_t(j)}\alpha_i(\{k \}) \leq B$ or $\beta_i^{\pi_t}(\{j\}) = \max\{B-\sum_{k: \pi_t(k) < \pi_t(j)}\alpha_i( \{k \}), 0\}$, if $\sum_{k: \pi_t(k) \leq \pi_t(j)}\alpha_i(\{k \}) > B$.
First, we are going to prove that for each $S \subseteq U,  \beta_i^{\pi_t}(S) \leq \min\{v(S),B\}, \forall i, \pi_t$. 

By the definition of $\beta_i^{\pi_t}$, it is clear to see that $ \beta_i^{\pi_t}(\{j\}) \leq \alpha_i(\{j\})$. Therefore, summing upon all items in $S$ (since we have additive functions), we get that:
$$ \beta_i^{\pi_t}(S) \leq \alpha_i(S) \leq \max_k \alpha_k(S) = v(S)$$
By the definition of $\beta_i^{\pi_t}$, we also have that $\beta_i^{\pi_t}(S)\leq B$.

Next, we are going to prove that for each $S \subseteq U: \exists \beta_i^{\pi_t}$ s.t. $\beta_i^{\pi_t}(S) = \min\{v(S),B\}$. We distinguish the following cases:
\begin{enumerate}
\item $v(S) \leq B$. Let $\pi_t$ be a permutation, s.t. all items in $S$ come first and let $\alpha_{i^*}$ be the maximizing function for set $S$, i.e. $v(S) = \alpha_{i^*}(S)$. Then, because $\sum_{j \in S} \alpha_{i^*}(\{j\}) \leq B$, we have $\beta_{i^*}^{\pi_t}(S) = \sum_{j \in S} \beta_{i^*}^{\pi_t}(\{j\}) = \sum_{j \in S} \alpha_{i^*}(\{j\}) = v(S)$.

\item $v(S) > B$. Let $\pi_t$ be a permutation, s.t. all items in $S$ come first and let $\alpha_{i^*}$ be the maximizing function for set $S$, i.e. $v(S) = \alpha_{i^*}(S)$.  Let $j^*$ be the last item in the permutation $\pi_t$ s.t. $\sum_{r:\pi_t(r) \leq \pi_t(j^*)} \alpha_{i^*}(\{r\}) \leq B$. Then, $\sum_{r:\pi_t(r) \leq \pi_t(j^*)} \beta_{i^*}^{\pi_t}(\{r\}) = \sum_{r:\pi_t(r) \leq \pi_t(j^*)} \alpha_{i^*}(\{r\})$. For the next items $z \in S$ in permutation $\pi_t$, we have $\beta_{i^*}^{\pi_t}(\{z\}) = \max \{B-\sum_{k: \pi_t(k) < \pi_t(z)}\alpha_i (\{k \}), 0\}$. In fact, the first item after $j^*$ will complete the missing value, in order to have: $\sum_{k:\pi_t(k) \leq \pi_t(j^*)+1} \beta_{i^*}^{\pi_t}(\{j\}) = B$ , and all subsequent items, $q$ will have $\beta_{i^*}^{\pi_t}(\{q\})=0$. Therefore, $\sum_{j \in S} \beta_{i^*}^{\pi_t}(\{j\}) = B$.
\end{enumerate}

\item (subadditive) Let $v$ be a monotone subadditive set function. We distinguish the following cases: 
\begin{enumerate}
\item If $\bv(S \cup T) = v(S \cup T) < B$. Then, we know for a fact that $\bv (S) = v(S) < B$ and that $\bv(T) = v(T) < B$. Then, $\bv (S \cup T) = v(S \cup T) \leq v(S) + v(T) = \bv (S) + \bv(T)$, where the inequality comes from the subadditivity of $v$.
\item If $\bv(S \cup T) = B < v(S \cup T)$. We have to further distinguish the following cases: 
\begin{enumerate}
\item $\bv(S) = B < v(S), \bv(T) = B < v(T)$. Then, $\bv(S \cup T) = B \leq 2B = \bv(S) + \bv(T)$.
\item $\bv(S) = B < v(S), \bv(T) = v(T) < B$. Then, $\bv(S \cup T) = B \leq B + v(T) = \bv(S) + \bv(T)$, where the inequality comes from the non-negativity of the liquid valuation.
\item $\bv(S) = v(S) < B, \bv(T) = B < v(T)$. Then, $\bv(S \cup T) = B \leq v(S) + B = \bv(S) + \bv(T)$, where the inequality again comes from the non-negativity of the liquid valuation.
\item $\bv(S) = v(S) < B, \bv(T) = v(T) < B$. Then, $\bv(S \cup T) = B \leq v(S \cup T) \leq v(S) + v(T) = \bv(S) + \bv(T)$, where the last inequality comes from the fact that $v$ is subadditive.
\end{enumerate}
\end{enumerate}
\end{itemize}
\end{proof}

\section{Supplementary Material for Section \ref{sec:worst-case}}\label{appendix:worst-case}

We include below the core theorems that are used in order to derive the $O(\sqrt{\log m})$-approximation to the LW, by adapting the techniques used by \citet{D16}.

\begin{proposition}[Strong Profitability of a set]\label{prop:strong-profit}
Let $S = \arg \max_{S' \subseteq U} \left\{v(S') - p(S') | p(S') \leq B \right\}$ a \emph{strongly profitable} set under item prices $p_1, \dots, p_m$ for valuation $v$. Then, $S$ is also a strongly profitable set for the liquid valuation $\bv$. 
\end{proposition}

\begin{proof}
Let $T \subseteq S$. We want to show that $\bv (T) \geq p(T)$. If $\bv(T) = v(T)$, then the property holds, since $S$ is strongly profitable for valuation $v$. If $\bv (T) = B$, then, due to monotonicity of $\bv$, $\bv (T) = \bv (S) \geq p(S) \geq p(T)$, where the first inequality comes from individual rationality.
\end{proof}

\begin{lemma}[Extension of Lemma 4.2 in \citep{D16}]\label{lem:4.2}
Let $\alpha = (\alpha_1, \dots, \alpha_n)$ be an allocation that is supported by prices $p_1, \dots, p_m$. A fixed price auction where budget constrained bidders make BCDQs and the items have prices $p_j' = \frac{p_j}{2}$ generates an allocation $\hat{A} = (\halpha_1, \dots, \halpha_n)$ with LW: $\sum_{i \in [n]} \bv_i(\halpha_i) \geq \frac{\sum_{j \in \cup_i \alpha_i} p_j}{4}$. 
\end{lemma}

\begin{proof}[Proof of Lemma \ref{lem:4.2}]
We will follow closely the proof presented by \citet{D16}, changing the analysis only slightly when it is required to reason about the set returned from the BCDQ. 

For every bidder $i$, let $W_i = \cup_{i' < i} \halpha_{i'}$ denote the set of competitive items that were allocated before bidder $i$ arrives to the auction. Let $\OPT_i = \sum_{j \in (\cup_{i' \geq i} \alpha_{i'}) \setminus W_i} p_j$. Then, $\OPT_1 = \sum_{j \in \cup_i \alpha_i} p_j$ and $\OPT_{n+1} = 0$. For every bidder $i \in [n]$ it holds that $W_{i+1} = W_i + \halpha_i$ and that the allocation $(\emptyset, \dots, \emptyset, \alpha_i \setminus W_i, \alpha_{i+1} \setminus W_i, \dots, \alpha_n \setminus W_i)$ is still supported by $p_1, \dots, p_m$. Thus, 
\begin{equation}\label{eq:1lem4.2}
\OPT_i - \OPT_{i+1} = \sum_{j \in (\alpha_i \setminus W_i)} p_j + \sum_{j \in \halpha_i} p_j
\end{equation}

Now notice that bidder $i$ could buy set $\alpha_i - W_i$ which implies that the liquid valuation that he got from the set that was ultimately received by the BCDQ was lower bounded by:
\begin{align*}
\bv_i (\alpha_i \setminus W_i) - \sum_{j \in (\alpha_i \setminus W_i)} q_j &= \bv_i(\alpha_i \setminus W_i) - \sum_{j \in (\alpha_i \setminus W_i)} \frac{p_j}{2} \\
&\geq \sum_{j \in (\alpha_i \setminus W_i)} \frac{p_j}{2} \numberthis{\label{eq:2lem4.2}} 
\end{align*}
Since the bidder had enough budget to buy set $\halpha_i$ (otherwise, it would not have been received as the answer of the BCDQ) we have that:
\begin{equation}\label{eq:3lem4.2}
\bv_i\left(\halpha_i \right) \geq \sum_{j \in \halpha_i}\frac{p_j}{2}
\end{equation}
Summing up Equations\eqref{eq:2lem4.2} and \eqref{eq:3lem4.2} and using Equation \eqref{eq:1lem4.2} we get:
$$2\bv_i\left(\halpha_i \right) \geq \sum_{j \in \left(\alpha_i \setminus W_i\right)} \frac{p_j}{2} + \sum_{j \in \halpha_i} \frac{p_j}{2} = \frac{\OPT_i - \OPT_{i+1}}{2}$$
which concludes our proof.
\end{proof}

\section{Supplementary Material for Section \ref{sec:lcma}.}\label{appendix:lcma}

\paragraph{Large Market Assumptions for \emph{Indivisible} Items.} Imagine, for example, a large market with $m$ indivisible items and $n$ bidders, s.t. $B_i \leq \frac{\bOPT}{m \cdot c}$ for some large constant $c > 1$. The number of bidders who receive \emph{at least} one item is \emph{at most} $m$ and therefore, $\bOPT \leq m \cdot B_{\max}$, which leads to $B_{\max} \leq B_{\max}/{c}$, which is a contradiction. We note here also that one can get similar voidness results for the case where $c < 1$; imagine a market with $n = m$ bidders and $m$ items, where the valuations of the bidders for the items are $v_{ii} = 1, v_{ij} = 0, j \neq i$ and $B_i \leq 1, \forall i \in [n]$. Then, the optimal LW is $\bOPT = mB_i$ (achieved when bidder $i$ gets item $i$). However, for any $c < 1$ it holds that $B_i \leq B_i/c$, while the market that we have in this example is a very thin market. 
 
In reality, the previous settings discussed in the literature possessed another crucial property, that made it possible for the large market assumption to enable the results about the constant factor approximation of the optimal LW. This property was the \emph{homogeneity} of the goods being auctioned; every bidder wanted exactly the same item or at least some portion of every item. The homogeneity of the goods, coupled with the large market assumption, essentially established \emph{competitive markets}.

\section{Supplementary Material for Section \ref{sec:stochastic}.} \label{appendix:stochastic}

\subsection{Missing Proofs}

\begin{proof}[Proof of Lemma \ref{lem:like3.4}]
We are going to follow the proof presented by \citet{FGL14}. For each $j \in S_i$, we denote by $\LW_j(\bbv) := A_i (\{j\})$ (i.e., $\LW_j(\bbv)$ corresponds to the contribution of item $j$ to the liquid welfare, under liquid valuation profile $\bbv$), where $A_i(\cdot)$ is the corresponding additive representative function for the set $S_i$. From the definition of $p_j$: 
\begin{align*}
p_j'    &= \E_{\bbv \sim \calD} \left[\LW_j(\bbv) - p_j' \right] + 2(p_j' - p_j) \numberthis{\label{eq:lem1}}\\
        &= \sum_{i \in [n]} \E_{\bbv \sim \calD} \left[\left(\LW_j(\bbv) - p_j' \right)\1 \left\{j \in S_i(\bbv) \right\} \right]  + 2(p_j' - p_j)   
\end{align*}

Let $\SOLD_i(\bbv, \pi)$ be the set of items that have been sold prior to the arrival of bidder $i$. Bidder $i$'s BCDQ receives set $S_i$ as the answer, from the items in $U \setminus \SOLD_i(\bbv, \pi)$ that maximizes $v(S_i) - p(S_i)$ subject to the fact that $p(S_i) \leq B_i$. Consider another random liquid valuation profile $\bbv'_{-i} \sim \calD_{-i}$, independent of $\bbv$. Let $S_i(\bv_i, \bbv_{-i}')$ be the allocation returned by $\ALG$ on input $(\bv_i, \bbv_{-i}')$. For the additive representative function $A_i$ for the set $S_i(\bv_i, \bbv_{-i}')$ it holds that $A_i(\{j\}) = \LW_j(\bv_i, \bbv_{-i}')$ for each $j \in S_i(\bv_i, \bbv_{-i}')$. Let $S_i(\bv_i, \bbv_{-i}, \bbv_{-i}') := S_i(\bv_i, \bbv_{-i}') \setminus \SOLD_i(\bbv, \pi)$ be the subset of items in $S_i(\bv_i, \bbv_{-i})$ that are available for purchase when bidder $i$ arrives. Since bidder $i$ could have bought set $S_i(\bv_i, \bbv_{-i}, \bbv_{-i}')$ but instead did not, using Lemma \ref{lem:bfdo} we get that: 
\begin{equation*}
2 \bv_i(S_i(\bbv)) - p(S_i(\bbv)) \geq \E_{\bbv_{-i}'} \left[\max \left\{ \LW_j(\bv_i, \bbv_{-i}') - p_j', 0 \right\} \right]
\end{equation*} 
Summing up for all the bidders and taking the expectation over all $\bbv \sim \calD$ we have: 
\begin{dmath}\label{eq:utility-bound1}
2 \E_{\bbv \sim \calD} \left[\sum_{i \in [n]} \bbv_i(S_i(\bbv)) \right] - \E_{\bbv \sim \calD} \left[\sum_{i \in [n]} p(S_i(\bbv)) \right] \geq \sum_{j \in U} \sum_{i \in [n]} \E_{\bv_i, \bbv_{-i}, \bbv_{-i}'} \left[\1 \left\{j \in S_i(\bv_i, \bbv_{-i}') \right\} \cdot \max \left\{ \LW_j(\bv_i, \bbv_{-i}') - p_j', 0  \right\} \cdot \1 \left\{j \neq \SOLD_i(\bbv, \pi) \right\} \right]  
\end{dmath}
Following exactly the same steps as in \citet{FGL14} we can rewrite the above as: 
\begin{align*}
2 &\E_{\bbv \sim \calD} \left[\sum_{i \in [n]} \bbv_i(S_i(\bbv)) \right] - \E_{\bbv \sim \calD} \left[\sum_{i \in [n]} p(S_i(\bbv)) \right] \\ 
&\geq \sum_{j \in U} \Pr_{\bbv} \left[j \neq \SOLD(\bbv, \pi) \right] \cdot (p_j + (p_j - p_j')) \numberthis{\label{eq:utility-bound2}} 
\end{align*}
For the expected revenue, due to individual rationality of the bidders it holds that: 
\begin{equation}\label{eq:rev-bound}
\E_{\bbv \sim \calD} \left[\Rev(\bbv, \pi) \right] = \sum_{j \in U} \Pr_{\bbv} \left[j \in \SOLD(\bbv,\pi) \right] \cdot (p_j - (p_j - p_j'))
\end{equation}
Adding Equations~\eqref{eq:rev-bound} and \eqref{eq:utility-bound2} we get:
\begin{align*}
2 &\E_{\bbv \sim \calD} \left[\sum_{i \in [n]} \bbv_i(S_i(\bbv)) \right] - \E_{\bbv \sim \calD} \left[\sum_{i \in [n]} p(S_i(\bbv)) \right] + \E_{\bbv \sim \calD} \left[\Rev(\bbv, \pi) \right] \\ 
&\geq \sum_{j \in U}p_j + \sum_{j \in U} (p_j - p_j') \left(1 - 2\Pr_{\bbv} \left[j \in \SOLD\left(\bbv, \pi\right) \right] \right) \\
&\geq \frac{1}{2}\E_{\bbv \sim \calD} \left[\sum_{i \in [n]} \bbv_i(S_i) \right] - \sum_{j \in U} \left|p_j - p_j' \right| \\ 
&\geq \frac{1}{2}\E_{\bbv \sim \calD} \left[\sum_{i \in [n]} \bbv_i(S_i) \right] - m\delta
\end{align*}
\end{proof}

\begin{proof}[Proof of Theorem \ref{thm:anon-price}]
Observe that we only needed to prove a variant of Lemma 3.4 by \citet{FGL14}, which we did in Lemma \ref{lem:like3.4}. This is due to the fact that the sampling arguments presented for finding the appropriate prices hold without any modification in our case too, since the liquid valuation remains XOS for XOS valuation functions. Further, no reasoning about incentives is required in this proof, since we are basing our arguments on ``ghost" samples that we draw from the known product distribution $\calD$.
\end{proof}

\end{document}